\documentclass[lettersize,journal]{IEEEtran}
	\usepackage{amsmath,amsfonts}
	\usepackage{amsmath}
	\usepackage{array}
	\usepackage{enumitem}
	\usepackage{textcomp}
	\usepackage{stfloats}
	\usepackage{url}
	\usepackage{verbatim}
	\usepackage{graphicx}
	\usepackage{cite}
	\usepackage{tikz}
	\usepackage{amsthm}
	\usepackage{subcaption}
	\usepackage[normalem]{ulem}
	\usepackage{etoolbox}
	\usepackage{flushend}
	\usepackage{amssymb}

	\usepackage{mathtools}
	\usepackage[mathscr]{euscript}
		
	\newtheorem{proposition}{Proposition}
	\usepackage{breqn}
	\usepackage[multiple]{footmisc}

	\definecolor{ForestGreen}{RGB}{34,139,34}
	\usepackage{bbm}
	\newtheorem{theorem}{Theorem}

	\DeclareMathOperator*{\argmax}{arg\,max}
	
	\allowdisplaybreaks
	
	\usepackage{algorithm}
	\usepackage{algpseudocode}
	\usepackage{setspace}
	\usepackage[font=footnotesize,labelfont=bf,skip=5pt]{caption}
	\usepackage{subcaption}
	\usepackage{bigints}
	\usepackage[colorlinks]{hyperref}
\hypersetup{
	colorlinks   = true, 
	urlcolor     = blue, 
	linkcolor    = blue, 
	citecolor   = {blue} 
}
	\usepackage[noabbrev,nameinlink]{cleveref}
	\usepackage{scalerel}[2016-12-29]
	\def\stretchint#1{\vcenter{\hbox{\stretchto[440]{\displaystyle\int}{#1}}}}
	
	\usepackage{wrapfig}
	
\begin{document}	
	\title{	\vspace{-0.4cm}\huge{Distributed IRSs Always Benefit Every Mobile Operator}}
	\author{L. Yashvanth,~\IEEEmembership{Student Member,~IEEE} and Chandra R. Murthy,~\IEEEmembership{Fellow,~IEEE}
	\thanks{This work was financially supported by the Qualcomm UR 6G India Grant, a research grant from MeitY, Govt. of India, and the Prime Minister's Research Fellowship, Govt. of India.}
			\thanks{The authors are with the Dept. of ECE, Indian Institute of Science, Bengaluru, India 560 012. (e-mails: \{yashvanthl, cmurthy\}@iisc.ac.in).}
		\vspace{-0.3cm}
		}
		\maketitle
		\begin{abstract}
We investigate the impact of multiple distributed intelligent reflecting surfaces (IRSs), which are deployed and optimized by a mobile operator (MO), on the performance of user equipments (UEs) served by other co-existing out-of-band (OOB) MOs that do not control the IRSs. We show that, under round-robin scheduling, in mmWave frequencies, the ergodic sum spectral efficiency (SE) of an OOB MO increases logarithmically in the total number of IRS elements with a pre-log factor that increases with the ratio of the number of OOB paths through the IRS to the number of elements at an IRS. We further show that the maximum achievable SE of the OOB MO scales log-linearly with the total IRS elements, with a pre-log factor of $1$. Then, we specify the minimum number of IRSs as a function of the channel parameters and design a distributed IRS system in which an OOB MO almost surely obtains the maximum SE. Finally, we prove that the outage probability at an OOB UE decreases exponentially as the number of IRSs increases, even though they are randomly configured from the OOB UE's viewpoint. We numerically verify our theory and conclude that distributed IRSs always help every MO, but the MO controlling the IRSs benefits the most.
\end{abstract}
		\section{Introduction}
		\IEEEPARstart{I}{\lowercase{ntelligent}} reflecting surfaces (IRSs) have been envisioned to meet the high data requirements for future wireless systems~\cite{RuiZhang_IRSSurvey_TCOM_2021,Choi_WCL_2021,Chen_TWC_2023}. They are made up of several passive elements that can be independently tuned to reflect signals in desired directions. Further, in a distributed version of it, multiple IRSs are uniformly spread out in the environment, which provides additional diversity benefits to the system. Also, in practice, multiple mobile operators (MOs) coexist in a geographical area and independently serve multiple user equipments (UEs) over non-overlapping frequency bands. Then, it is important to understand how a distributed IRS  deployed and optimized by only one of the MOs affects the performance of other out-of-band (OOB) MOs. This aspect is pertinent because an IRS does not have a band-pass filter and reflects every signal that impinges it over a wide bandwidth. This paper analyzes the impact of distributed IRSs on other MOs in the mmWave band.

		 A few works in the literature have shown that distributed IRSs improve multiplexing gains in multiple antenna  systems~\cite{Choi_WCL_2021,Chen_TWC_2023}. Similarly, they have been used to reduce link blockages in ultra-reliable and low latency communication (URLLC) applications~\cite{Zhang_ArXiv_2023_TWC_EA}. Further, multiple IRSs can be used to improve coverage in cellular systems~\cite{Shi_TCOM_2023}. Also, distributed IRSs have been leveraged to efficiently mitigate multi-user interference in multi-user scenarios~\cite{Peng_TGCN_2022}.  Finally, these merits can be obtained with minimal channel estimation overheads~\cite{Yashvanth_Globecom_2023}.  
		 
		 The above-existing works implicitly assume that a single MO controls the IRSs. However, if multiple MOs co-exist and only one of them deploys multiple IRSs to serve its UEs optimally, then the impact of the IRSs on the performances of the OOB MOs (which do not control the IRSs) is unexplored but is an important aspect to be understood in real-world deployment scenarios. Further, BSs of different MOs typically do not cooperate with each other. Hence, existing approaches that jointly optimize the IRS phases by cooperation across BSs cannot be used to solve this problem. Although we studied this aspect in a single IRS case \cite{Yashvanth_TCOM_2023}, the channel properties in distributed IRSs are different due to multiple independent links offered by the IRSs. This leads to new results \& insights, so it merits an independent study for the multiple IRS scenario.
		 
		We consider two MOs, X and Y, operating on non-overlapping mmWave bands. The (in-band) MO X deploys and optimizes multiple distributed IRSs to serve a UE in every time slot.  The (OOB) MO Y does not deploy any IRS and is oblivious to MO X's IRSs.\footnote{The broad conclusions of this paper can be shown to hold for any number of OOB MOs and also when every MO has its own set of IRSs.}  
		Then, our key contributions are:		 
		\begin{enumerate}[leftmargin=*]
		 \item Under round-robin (RR) scheduling, we derive the ergodic sum spectral efficiencies (SE) of the MOs. If $N$ is the total number of IRS elements, we show that the SE of MO X grows as $\!\mathcal{O}(2\log_2(N))$, and the SE of MO Y scales as $\mathcal{O}(\tau\log_2(N))$, where the pre-log factor $\tau\!\in\! [0,1]$  increases with the ratio of the number of OOB paths through the IRS to the number of elements at an IRS. (See Theorem~\ref{thm:rate_characterization_mmwave_single_path_IB_multiple_IRS}.)
		 		 \item We design a distributed IRS system and specify the minimum number of IRSs for MO Y to \emph{almost surely} achieve the maximum SE (i.e., for $\tau=1$.) (See Proposition~\ref{prop_OOB_SNR_linear_N}.)
		 \item Finally, we show that the outage probability at an arbitrary OOB UE decreases exponentially as the number of IRSs deployed by MO X increases. (See Theorem~\ref{thm_outage_prob_OOB_UE}.)
		 \end{enumerate}
		 Through numerical simulations, we affirm that distributed IRSs enhance the performance of not only the UEs for which it is optimized but also of other OOB UEs at no additional signal processing costs, both instantaneously and on average.\\
		\indent  \emph{Notation:} $\mathcal{CN}(0,\sigma^2), \text{Ber}(p), \text{Bin}(n,p)$ are circular symmetric Gaussian, Bernoulli, and Binomial distributions with parameters $\sigma^2$, $p$, and $n,p$, respectively. $\odot$ is the Hadamard product. $\mathsf{Pr}(\cdot)$ and $\mathbb{E}[\cdot]$ stand for probability and expectation, $\mathcal{O}$ is the Landau's Big-O notation, $x!$,  $\lceil \cdot \rceil$, $\mathbbm{1}_{\{\cdot\}}$, and $\cup$ refer to the factorial of $x$, ceil \& indicator functions, and set unions.
		\section{System Model and Problem Statement}
		We consider a wireless system where two MOs, X and Y, provide services to $K$ and $Q$ UEs on non-overlapping mmWave frequency bands using their base stations: BS-X and BS-Y, as shown in Fig.~\ref{fig:Network_scenario_single_IRS}. For simplicity and to isolate the \emph{impact of IRSs} on the system, we consider a single antenna BSs~\cite{Yaqiong_JSTSP_2022}. However, all our results directly extend to multiple antenna BSs. MO X deploys and optimizes $S>1$ distributed IRSs with $M$ elements each (with a total of $N=SM$ IRS elements) to optimally serve its (in-band) UEs in every time slot. On the other hand, MO Y does not have an IRS and is oblivious to the presence of MO X's IRSs. 
			\begin{figure}[t!]
				\vspace{-0.3cm}
			\centering
			\includegraphics[width=0.78\linewidth]{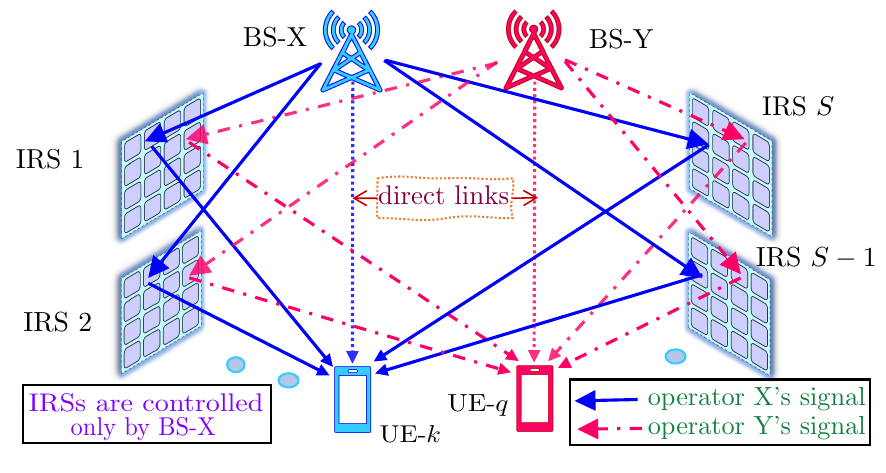}
			\caption{Network scenario of a distributed IRS-aided two-operator system.}
			\label{fig:Network_scenario_single_IRS}
			\vspace{-0.15cm}
			\vspace{-1\baselineskip}
			\end{figure}
		The downlink signal received at the $k$th (in-band) UE served by BS-X is\footnote{We neglect inter-IRS reflections; they experience higher path loss and contribute to negligible received energy compared to single IRS reflections.}\!
			\vspace{-0.1cm}
		\begin{equation}\label{eq:airtel_downlink}
				\vspace{-0.1cm}
			\!\!y_{k} = \left(h_{d,k}+\sum\nolimits_{s=1}^S \mathbf{g}_{s,k}^T\boldsymbol{\Theta}_s\mathbf{f}_s^X\right)x_k + n_k ,
		\end{equation} where $\{\mathbf{g}_{s,k},\mathbf{f}_s^X\} \in \mathbb{C}^{M \times 1}$ are the channels from $s$th IRS to $k$th UE and BS-X to $s$th IRS, respectively; the diagonal matrix $\boldsymbol{\Theta}_s$ contains the phase shifts of the elements at $s$th IRS, i.e., 
					\vspace{-0.05cm}
		\begin{equation}
					\vspace{-0.05cm}
		\boldsymbol{\Theta}_s \triangleq \text{diag}\left(\left[e^{j\zeta_{1,s}},e^{j\zeta_{2,s}},\ldots,e^{j\zeta_{M,s}}\right]\right) \in \mathbb{C}^{M \times M},
		\end{equation}
		 where $\zeta_{m,s}$ is the phase shift at $m$th element of $s$th IRS, $h_{d,k}$ is the direct channel from BS to UE-$k$. Also, $x_k$ is UE-$k$'s data symbol with $\mathbb{E}[|x_k|^2] = P$, and $n_k$ is the additive noise with $n_k \sim \mathcal{CN}(0,\sigma^2)$. Similarly, at OOB UE-$q$ served by BS-Y,
		\vspace{-0.1cm}
		\begin{equation}\label{eq:jio_downlink}
				\vspace{-0.1cm}
			y_{q} = \left(h_{d,q}+\sum\nolimits_{s=1}^{S} \mathbf{g}_{s,q}^T\boldsymbol{\Theta}_s\mathbf{f}_s^Y\right)x_q	+ n_q.
		\end{equation} 
	We consider the Saleh-Valenzuela channel model~\cite{Wang_TWC_2022_IRS_BA}:
			\vspace{-0.15cm}
		\begin{equation}\label{eq:ch_model_mmwave_single_IRS}
				\vspace{-0.1cm}
			\mathbf{f}_s^{p} \!=\!\! \sqrt{\!\tfrac{M}{L^{(1)}_{s,p}}}\sum_{i=1}^{L^{(1)}_{s,p}}\!\!\gamma^{(1)}_{i,s,p} \mathbf{a}^*_M(\phi^{i}_{s,p}); \mathbf{g}_{s,t} \!=\!\! \sqrt{\!\tfrac{M}{L^{(2)}_{s,t}}}\sum_{j=1}^{L^{(2)}_{s,t}}\!\!\gamma^{(2)}_{j,s,t}\mathbf{a}^*_M(\psi^{j}_{s,t}),\!\!\!
		\end{equation} where $p \in \{X,Y\}$, $L^{(1)}_{s,p}$ and $L^{(2)}_{s,t}$ are the number of resolvable spatial paths in the BS-$p$ to $s$th IRS, and $s$th IRS to UE-$t$ links, respectively. 
		For notational simplicity, we let $L^{(1)}_{s, X}=L^{(1)}_{s, Y}\triangleq L_1$, and $L^{(2)}_{s,t}= L_2 $ $\forall t, s$. Further, $\phi^{i}_{s,p}$ and $\psi^{j}_{s,t}$ denote the sines of the angle of arrival of the signal from BS-$p$ at $s$th IRS in the $i$th path, and the departure from $s$th IRS to the $t$th UE in the $j$th path, respectively, where sine of an angle ($\phi_\chi$) is related to the physical angle ($\chi$) by\footnote{In the sequel, the term ``angle" will denote the sine of a physical angle.}
	$\phi_\chi \!=\!  (2d/\lambda)\sin(\chi)$. Here, $d$ and $\lambda$ represent the IRS inter-element distance and the signal wavelength, respectively. The fading coefficients, $h_{d,t},\gamma^{(1)}_{i,s,p}$ and $\gamma^{(2)}_{j,s,t}$, are independently sampled from $\mathcal{CN}(0,\beta_{d,t}), \mathcal{CN}(0,\beta_{\mathbf{f}^p})$, and $\mathcal{CN}(0,\beta_{\mathbf{g},t})$, respectively, where $\beta$'s are the link path losses.\footnote{We consider the path losses $\beta_{\mathbf{f}^p}$ and $\beta_{\mathbf{g},t}$ to be independent of the IRS index~\cite{Lin_Bai_Access_2017}. In practice, the path losses of these links depend on the exact IRS locations. Optimizing the IRS locations is beyond the scope of this paper.} The array response vector of a uniform linear array (ULA) based IRS is denoted by $\mathbf{a}_M(\phi)$, with angle $\phi$ (and $d\!=\!\lambda/2$)\footnote{The results in the paper can be extended to other types of IRS geometries. We consider the ULA geometry for simplicity~\cite{FeiZhao_WCL_2023}.}:
		\vspace{-0.2cm}
		\begin{equation}\label{eq:array_vector_template}
				\vspace{-0.2cm}
			\mathbf{a}_M(\phi) =1/{\sqrt{M}} \left[1, e^{-j\pi\phi},\ldots,e^{-j(M-1)\pi\phi}\right]^T.
		\end{equation}
Recall that an $M$-element IRS forms at most $M$ resolvable beams~\cite{Wang_TWC_2022_IRS_BA,Rui_Zhang_WCL_2020}. So, the path angles in~\eqref{eq:ch_model_mmwave_single_IRS} are drawn independently and uniformly at random from the angle-book, $\mathbf{\Phi} \triangleq\! \left\{\! \left(\!-1 \!+ \!\frac{2i}{M}\right)\!\big\rvert i =  0,\ldots,M-1\!\right\}$.  Further, since only one of the $L \triangleq L_1L_2$ cascaded paths contains most of the energy, aligning the IRS along this path procures near-optimal benefits at the in-band UEs. So, the overall channels to the in-band UEs can be simplified by their dominant paths~\cite{Jinghe_TCOM_2023}. Using~\eqref{eq:ch_model_mmwave_single_IRS} in~\eqref{eq:airtel_downlink} with $L_1=L_2=1$, the effective channel for UE-$k$ is, $h_k$ 
 \begin{align}
 \!\!&\!\!= h_{d,k}\!+\!M\sum\nolimits_{s=1}^{S} \gamma^{(1)}_{1,s,X}\gamma^{(2)}_{1,s,k}\mathbf{a}_M^H(\psi^{1}_{s,k})\boldsymbol{\Theta}_s\mathbf{a}_M^*(\phi^{1}_{s,X}) \!\!\!\\
&\stackrel{(a)}{=} h_{d,k}  \!+ \!M\sum_{s=1}^{S}  \left(\! \gamma^{(1)}_{1,s,X}\gamma^{(2)}_{1,s,k}\left(\mathbf{a}_M^H(\phi^{1}_{s,X})\!\odot\mathbf{a}_M^H(\psi^{1}_{s,k})\right)\right)\!\boldsymbol{\theta}_s  \nonumber\\
&\stackrel{(b)}{=} h_{d,k} + M \sum\nolimits_{s=1}^{S} \gamma_{X,s,k}\mathbf{\dot{a}}_M^H(\omega^1_{X,s,k})\boldsymbol{\theta}_s, \label{eq:ch_model_IB_UE}
\end{align}
where $\boldsymbol{\theta}_s \!\triangleq\! \text{diag}(\boldsymbol{\Theta}_s)$, $(a)$ is due to  the properties of Hadamard products. In $(b)$, \textcolor{black}{$\omega^1_{X,s,k} \triangleq \sin^{-1}\left(\sin(\phi^{1}_{s,X}) + \sin(\psi^{1}_{s,k})\right)$}, $\gamma_{X,s,k} \triangleq \gamma^{(1)}_{1,s,X}\gamma^{(2)}_{1,s,k}$, and $\mathbf{\dot{a}}_M(\cdot)$ is the array vector normalized by $M$ (see~\eqref{eq:array_vector_template}), i.e., $\mathbf{\dot{a}}_M(\cdot) = \frac{1}{\sqrt{M}}\mathbf{a}_M(\cdot)$. Also, since the IRS is not optimized to the OOB UEs, we retain the  channel at an OOB UE-$q$ with all the paths and simplify it as
\vspace{-0.05cm}
\begin{equation}\label{eq_OOB_channel}
h_q = h_{d,q} +  \frac{M}{\sqrt{L}}\sum\nolimits_{s=1}^{S} \sum\nolimits_{\ell=1}^{L} \gamma^{\ell}_{Y,s,q}\mathbf{\dot{a}}_M^H(\omega^{\ell}_{Y,s,q})\boldsymbol{\theta}_s,
\vspace{-0.1cm}
\end{equation} 
where $L\!\triangleq\! L_{s,Y}^{(1)}L_{s,q}^{(2)}$, $\{\gamma^{\ell}_{Y,s,q}\}_{\ell=1}^{L} \!\triangleq\! \{\gamma^{(1)}_{i,s,Y}\gamma^{(2)}_{j,s,q}\}_{i=1,j=1}^{L^{(1)}_{s,Y},L^{(2)}_{s,q}} \ \forall s$.\\
\vspace{-0.2cm}

We can now mathematically state our problem. Suppose the BS-X configures all the $S$ IRSs to maximize the SE at UE-$k$ by solving the joint optimization problem (across the IRSs):
\vspace{-0.2cm}
	\begin{equation}\label{eq_opt_prob}
\!\!\!\!\!\left\{\boldsymbol{\Theta}_s^{\mathrm{opt}}\right\}_{s=1}^{S} \!\!=\!\argmax_{\left\{\boldsymbol{\Theta}_s \right\}_{s=1}^S} \log_2\!\!\left(\!\!1\!\!+\!\left|h_{d,k}\!+\!\!\sum\limits_{s=1}^S \mathbf{g}_{s,k}^T\boldsymbol{\Theta}_s\mathbf{f}_s^{X} \right|^2\!\!\!\!\frac{P}{\sigma^2}\!\!\right)\!,\!\!\!	\tag{P1}
\end{equation} 
subject to $\boldsymbol{\Theta}_s \!\in \!\mathbb{C}^{M \times M}$ being diagonal matrices with unit modulus entries. Then, these phases are randomly tuned from an OOB UE's viewpoint. This is because different MOs do not coordinate with each other, and so it is not feasible to configure the IRSs that are jointly optimal to UEs served by both BS-X and BS-Y. In this context, we address the following:
\begin{enumerate}[leftmargin=*]
\item What is the effect of the randomly configured IRSs on the ergodic SE of the UE served by the OOB MO Y? 
\item What is the best ergodic sum-SE that the OOB MO Y can obtain, and when is it achievable?  
\item How does the outage probability of (OOB) UE-$q$ scale with $S$, $M$, and the channel parameters? 
\end{enumerate}
We answer the above questions in the following sections.
		\section{Ergodic Sum Spectral Efficiency Analysis}\label{sec_single_IRS}
	We begin by noting that the optimal IRS configuration at the $s$th IRS, i.e., the solution to~\eqref{eq_opt_prob}, can be obtained as
	\vspace{-0.05cm}
	\begin{equation}\label{eq:opt_IB_mmwave_airtel_multiple_IRS}
	\boldsymbol{\theta}_s^{\mathrm{opt}} = \dfrac{h_{d,k}\gamma^*_{X,s,k}}{\left|h_{d,k}\gamma_{X,s,k} \right|} \times M \mathbf{\dot{a}}_M(\omega^1_{X,s,k}).
	\vspace{-0.05cm}
	\end{equation}
Clearly, $\boldsymbol{\theta}_s^{\mathrm{opt}}$ has nonzero response in the direction of the in-band UE-$k$'s channel via the $s$th IRS~\cite[Fig.~$3$]{Yashvanth_TCOM_2023}. However, since the OOB channels through the IRSs are also directional, with nonzero probability, one or more of the IRSs also align to an OOB UE's channel. Specifically, for single IRS with a flat-top beamforming pattern~\cite{Madhow_TACM_2011}, with probability $\bar{L}/M$ ($\bar{L}\triangleq\min{\{L, M\}}$), the IRS aligns with the OOB UE, and with probability $\!1-\bar{L}/M\!$, it does not align to the OOB UE, following a Bernoulli distribution~\cite[Proof of Theorem~$3$]{Yashvanth_TCOM_2023}. Now, with $S$ distributed IRSs, since the beam patterns are independent across IRSs and the channels at the OOB UE via each IRS are independent, the overall beamforming pattern at the OOB UE follows the distribution of the sum of $S$ independent and identically distributed Bernoulli random variables, i.e., the Binomial distribution.
 In this view, we next characterize the sum-SE of MOs under round-robin (RR) scheduling of UEs.\!\!\!
 \vspace{-0.1cm}
 		\begin{theorem} \label{thm:rate_characterization_mmwave_single_path_IB_multiple_IRS}
	Consider a distributed IRS-aided mmWave system consisting of $S$ non-colocated IRSs, each with $M$ elements. Then, if the IRSs are optimized (as per~\eqref{eq:opt_IB_mmwave_airtel_multiple_IRS}) to serve the UEs of MO X in every time slot, the ergodic sum-SEs of MOs X and Y, $\bar{S}_M^{(X)} $ and $\bar{S}_M^{(Y)} $ under RR scheduling, scale~as~\eqref{eq:mmwave_rate_airtel_rr_multiple_path} and~\eqref{eq:mmwave_rate_jio_rr_multiple_path} at the top of this page, respectively, where $N$ is the total number of IRS elements ($N=SM$), and $\eta \triangleq M/N$.\!\! 
		 \vspace{-0.2cm}
	\end{theorem}	
			\begin{figure*}[t]
			\vspace{-0.5cm}
	\begin{align}\label{eq:mmwave_rate_airtel_rr_multiple_path}
	\bar{S}_M^{(X)}\! \approx \!\frac{1}{K}\sum_{k=1}^K \log_{2}\left(1 + \left[N^2\left(\dfrac{\pi^2}{16} + \eta \left(1-\dfrac{\pi^2}{16}\right)\right)\beta_{r,k} + N \frac{\pi^{3/2}}{4}\sqrt{\beta_{d,k}\beta_{r,k}}+ \beta_{d,k}\right]\frac{P}{\sigma^2} \right), 	\end{align}
	\vspace{-0.3cm}
		\rule{\textwidth}{0.3mm}
	\end{figure*}
	
	\begin{figure*}[t]
	\vspace{-0.3cm}
	\begin{align}\label{eq:mmwave_rate_jio_rr_multiple_path}
		\!\!\!\!\bar{S}_M^{(Y)} \approx \begin{cases}
			\dfrac{1}{Q}\sum\limits_{q=1}^{Q}  \sum\limits_{\mathfrak{s}=0}^{S} \dfrac{S!}{\left(S-\mathfrak{s}\right)!\mathfrak{s}!}
			{\left(\dfrac{L}{M}\right)}^{\mathfrak{s}}{\left(1-\dfrac{L}{M}\right)}^{S-\mathfrak{s}}\log_2\left(1 + \left[\dfrac{\mathfrak{s}{M}^2}{L}\beta_{r,q}+\beta_{d,q}\right]\dfrac{P}{\sigma^2}\right), & \mathrm{ if } \ \  L < M, \\
			\dfrac{1}{Q}\sum\limits_{q=1}^{Q}\log_2\left(1+\left(\beta_{d,q} + N\beta_{r,q}\right)\dfrac{P}{\sigma^2}\right),& \mathrm{ if } \ \  L \geq M.
		\vspace{-0.2cm}
		\end{cases}
	\end{align}
	\vspace{-0.3cm}
	\rule{\textwidth}{0.3mm}
	\vspace{-0.7cm}
\end{figure*}
\vspace{-0.1cm}
\begin{proof}
We prove the theorem for MOs X and Y separately.
\subsubsection{Ergodic sum-SE of MO-X}

Let the ergodic SE at an arbitrary UE served by MO X, say $k$, be $\langle S_M^{k, X} \rangle$. Then, the ergodic sum-SE of MO X under RR scheduling is $\bar{S}_M^{(X)} = \frac{1}{K} \sum_{k=1}^{K} \langle S_M^{k, X} \rangle$. Computing the exact expression for $\langle S_M^{k, X} \rangle $ leads to intractable forms and is not insightful. So, we apply the Jensen's approximation to $\langle S_M^{k, X} \rangle$ and obtain  
			\begin{equation}
			\!\!\langle S_M^{k, X} \rangle\! = \!\mathbb{E}\!\left[\log_2\! \!\left(\!1\!+\!|h_{k}|^2 \!\frac{P}{\sigma^2}\!\right)\!\right]\! \approx\! \log_2\!\!\left(\!1\!+\!\mathbb{E}\!\left[|h_{k}|^2\right] \! \frac{P}{\sigma^2}\!\right),\!\!
			\end{equation}
			where the expectations are taken with respect to the channels of UE-$k$. To evaluate $|h_k|^2$, we use~\eqref{eq:opt_IB_mmwave_airtel_multiple_IRS}  in~\eqref{eq:ch_model_IB_UE}, and obtain the overall channel gain at UE-$k$ in~\eqref{eq:optimal_ch_UE_k}. 
			\begin{figure*}
			\begin{align}\label{eq:optimal_ch_UE_k}
			|h_k|^2 =\left| |h_{d,k}| + M \sum\nolimits_{s=1}^{S} |\gamma_{X,s,k}|\right|^2  = |h_{d,k}|^2 + M^2 {\left(\sum\nolimits_{s=1}^S |\gamma_{X,s,k}| \right)}^{2}+ 2M|h_{d,k}|\sum\nolimits_{s=1}^S |\gamma_{X,s,k}|.
			\end{align} 
			\vspace{-0.3cm}
	\rule{\textwidth}{0.3mm}
	\vspace{-0.6cm}
			\end{figure*}
			Next, we evaluate the expected values of the three terms in~\eqref{eq:optimal_ch_UE_k}, below:\\
			\underline{\emph{Term I:}} It is clear that $\mathbb{E}[|h_{d,k}|^2] = \beta_{d,k}$. Let $\beta_{r,k} \triangleq \beta_{\mathbf{f}^X}\beta_{\mathbf{g},k}$.\\
			\underline{\emph{Term II:}} $\mathbb{E}\left[\left(\sum\limits_{s=1}^S |\gamma_{X,s,k}| \right)^2\right] = \mathbb{E}\left[\mathop{\sum\limits_{s,p=1}^{S}}_{s \neq p} |\gamma_{X,s,k}| |\gamma_{X,p,k}| \right. $ $\left. + \sum\nolimits_{s=1}^{S}|\gamma_{X,s,k}|^2\right]  =  S(S-1)\left(\frac{\pi}{4}\sqrt{\beta_{r,k}}\right)^2 + S\beta_{r,k}$. Collecting and re-arranging the factors, the expected value of $M^2\! {\left(\sum\nolimits_{s=1}^S |\gamma_{X,s,k}| \!\right)}^{\!2}\!$ is $M^2\beta_{r,k}\left[S^2 \frac{\pi^2}{16} + S\left(1-\frac{\pi^2}{16}\right)\right].$\!\\
			\underline{\emph{Term III:}} Finally, we can show $2M\mathbb{E}\left[|h_{d,k}|\sum_{s=1}^S |\gamma_{X,s,k}|\right]=MS\frac{\pi^{3/2}}{4}\sqrt{\beta_{r,k}\beta_{d,k}}$, due to   independence of $h_{d,k}$ and $\gamma_{X,s,k}$. Summing the above terms and using it in $\bar{S}_M^{(X)}$ yields~\eqref{eq:mmwave_rate_airtel_rr_multiple_path}. 
 
\subsubsection{Ergodic sum-SE of MO-Y} Let the ergodic SE at an arbitrary OOB UE, say $q$, be $\langle S_{M}^{q, Y} \rangle$. Then, under RR scheduling, the ergodic sum-SE is $\bar{S}_M^{(Y)} = \frac{1}{Q}\sum_{q=1}^{Q} \langle S_{M}^{q,Y}  \rangle$. 
First, we prove~\eqref{eq:mmwave_rate_jio_rr_multiple_path} for $L < M$. Define the random variable $A_{\mathrm{s}}$ to denote whether the $\mathrm{s}$th IRS aligns with the channel to OOB UE-$q$ or not; then $A_{\mathrm{s}}\! \sim\! \text{Ber}(L/M)$ for all $\mathrm{s}$. Further, let $B$ count the number of IRSs aligning to the channel of UE-$q$; then, $B = \sum_{\mathrm{s}=1}^{S}A_{\mathrm{s}}$. So, $B \!\sim\! \text{Bin}(S,L/M)$. Now, $\langle S_{M}^{q,Y} \rangle =$\!\!
\begin{align}
&= \mathbb{E}_{|h_q|^2,B}\!\left[\log_2\left(1+\frac{|h_q|^2 P}{\sigma^2}\right)\mathbbm{1}_{\left\{B \in \{0,1,2,\ldots,S \} \!\right\}}\right]\!\\
&\stackrel{(a)}{=}\sum_{\mathfrak{s}=0}^{S} \mathbb{E}
_{|h_q|^2}\!\!\left[\log_2\!\!\left(1+\frac{|h_q|^2 P}{\sigma^2}\right)\!\!\Bigg\vert B=\mathfrak{s}\right] \!{\sf {Pr}}(B=\mathfrak{s})\\
& \stackrel{(b)}{\approx} \sum_{\mathfrak{s}=0}^{S} \log_2\left(1+\mathbb{E}
\left[|h_q|^2\Big\vert B=\mathfrak{s}\right] \frac{P}{\sigma^2}\right){\sf {Pr}}(B=\mathfrak{s}) \label{eq_ergodic_SE_formula_UE_q},
\end{align} 
where $(a)$ is due to the law of iterated expectations, and (b) is due to Jensen's approximation, which is known to be numerically tight. Note that ${\sf {Pr}}(B=\mathfrak{s}) = \frac{S!}{(S-\mathfrak{s})!\mathfrak{s}!}{\left(\frac{L}{M}\right)}^{\mathfrak{s}}{\left(1\!-\!\frac{L}{M}\right)}^{S-\mathfrak{s}}$. Further, $\mathbb{E}
\left[|h_q|^2\Big\vert B=\mathfrak{s}\right] $ is the average channel gain at UE-$q$ when the beams from exactly $\mathfrak{s}$ of the IRSs align with one of the $L$ paths to UE-$q$ through these IRSs. Then, the channel gain due to these \emph{contributing} IRSs is computed as (see~\eqref{eq_OOB_channel})
\vspace{-0.1cm}
\begin{equation}
|h_q|^2 \Big\vert \left\{B=\mathfrak{s}\right\}= \left|h_{d,q} +  \frac{M}{\sqrt{L}}\sum\nolimits_{s=1}^{\mathfrak{s}} \gamma^{\ell^*}_{Y,s,q}\right|^2,
\vspace{-0.1cm}
\end{equation}
 where $\gamma^{\ell^*}_{Y,s,q}$ is the channel gain of the cascaded path that aligns with the $s$th IRS. We can now show that $\mathbb{E}
\left[|h_q|^2\Big\vert B=\mathfrak{s}\right] = (\mathfrak{s}{M}^2/L)\beta_{r,q}+\beta_{d,q}$, where $\beta_{r,q} \triangleq  \beta_{\mathbf{f}^Y}\beta_{\mathbf{g},q}$. Collecting all the terms for~\eqref{eq_ergodic_SE_formula_UE_q} and plugging into $\bar{S}_M^{(Y)}$ yields~\eqref{eq:mmwave_rate_jio_rr_multiple_path} for $L < M$. Now, for $L \geq M$, the probability term $\bar{L}/M=1$. So, every IRS almost surely aligns with the OOB UE. Thus, the overall channel gain and its mean are
\vspace{-0.1cm}
\begin{equation}\label{eqn_OOB_ch_gain_large_L}
\vspace{-0.1cm}
|h_q|^2 = \left|h_{d,q} +  \sqrt{M}\sum\nolimits_{s=1}^{S} \gamma^{\ell^*}_{Y,s,q}\right|^2, \text{ and }
\end{equation} $N\beta_{r,q}\!+\!\beta_{d,q}$, respectively. The proof can now be finished.
\end{proof}
\vspace{-0.3cm}
\noindent We observe from Theorem~\ref{thm:rate_characterization_mmwave_single_path_IB_multiple_IRS} that the ergodic sum-SE of even an OOB MO monotonically grows in $N$, with the peak scaling of $\mathcal{O}(\log_2(N))$ when $L \geq M$. For $L\!<\!M$, as we show in Sec.~\ref{sec_numerical_results}, the OOB SE scales as $\mathcal{O}(\tau\log_2(N))$ for $\tau \in [0,1)$. Here, the exact value of $\tau$ depends on how $L$ compares with $M$. The primary reason for this improvement in the OOB SE is that the IRSs create more paths that enrich the channel at the OOB UEs. We note that the above results generalize the work of~\cite{Yashvanth_TCOM_2023}, which considers $S=1$.
Further, the SE at the (in-band) MO X still scales as $\mathcal{O}(2\log_2(N))$, similar to single IRS setups. Thus, deploying multiple IRSs in a distributed manner retains optimal benefits at the in-band MO and helps other MOs simultaneously at no significant overheads.\!\! 
\vspace{-0.2cm}
\subsection{Design for achieving $\mathcal{O}(\log_2(N))$ growth in the OOB SE}
		From~\eqref{eq:mmwave_rate_jio_rr_multiple_path}, it is clear that the OOB SE maximally scales log-linearly in the number of IRS elements and is achieved when the number of paths $L$ is sufficiently large. Specifically, if $L\geq M$, every IRS contributes to the signal at an OOB UE. We leverage this fact to design new specifications for distributed IRS systems in the following result to almost surely procure $\mathcal{O}(\log_2(N))$ growth in the OOB SE for any $L, N$.
		\begin{proposition}\label{prop_OOB_SNR_linear_N}
		Consider a system where an MO deploys and controls multiple IRSs in a distributed fashion, using a total of $N$ IRS elements. If the number of IRSs $S$ each with $M$ elements providing $L$ cascaded paths at an OOB UE, satisfy 				\begin{equation}\label{eq_linear_gain_system}
		M \leq M^* \triangleq N^{\delta^*}, \text{ and, } S \geq S^* \triangleq \left \lceil N^{1-\delta^*} \right \rceil, 
		\end{equation}
		where $\delta^* = \min\{1,\log_N L\}$, then almost surely, the OOB SE attains the maximum scaling of $\mathcal{O}(\log_2(N))$ for any $N, L$. 
		\end{proposition}
				\vspace{-0.1cm}
		\begin{proof}
		If $S$ and $M$ satisfy~\eqref{eq_linear_gain_system}, we have $L/M \geq 1$ under any scenario. Thus, almost surely, every IRS contributes to the signal at an OOB UE. So, the overall channel at the OOB UE is~\eqref{eqn_OOB_ch_gain_large_L}, for which the ergodic SE scales as $\mathcal{O}(\log_2(N))$.
		\end{proof}
		\vspace{-0.1cm}
		Proposition~\ref{prop_OOB_SNR_linear_N} says that, if the number of paths is $L=N^{\delta}$ for $ 0\!\leq\!\delta\!\leq\!1$, then a distributed system designed as per~\eqref{eq_linear_gain_system} gives maximum possible benefits at the OOB MO. Further, this does not affect the optimal growth of the SE for MO X. \\
		\vspace{-0.35cm}
		
		Next, we analyze the instantaneous channel characteristics at an OOB UE due to these arbitrarily configured IRSs.
		\section{Outage Probability Evaluation}
		We now examine the outage probability of an OOB UE-$q$ as a function of the IRS and channel parameters. The outage probability is $P_{\text{out},q}^{\rho} \triangleq \mathsf{Pr}(|h_q|^2 \leq \rho)$, which from~\eqref{eq_OOB_channel} becomes
				\vspace{-0.1cm}
		\begin{equation}
		\!\!\!P_{\text{out},q}^{\rho}\! = \!\mathsf{Pr}\!\left(\!\left|h_{d,q}\! + \! \frac{M}{\sqrt{L}}\sum\limits_{s=1}^{S} \!\sum\limits_{\ell=1}^{L}\!\! \gamma^{\ell}_{Y,s,q}\mathbf{\dot{a}}_M^H(\omega^{\ell}_{Y,s,q})\boldsymbol{\theta}_s\right|^2\!\!\!\!\leq\! \rho\!\right)\!\!.\!\!
				\vspace{-0.1cm}
		\end{equation}
		Intuitively, each IRS provides an independent link at the OOB UE, so multiple IRSs should provide better diversity gains than single or no IRS scenarios. We have the following result. 
		\begin{theorem}\label{thm_outage_prob_OOB_UE}
		The outage probability at an OOB UE-$q$ due to $S$ randomly configured and distributed IRSs, each with $M$ elements and contributing to $L$ paths, is given by 
		\begin{equation}\label{eqn_outage_prob_OOB_UE}
		\vspace{-0.2cm}
		P_{\emph{out},q}^{\rho} = \left(1- e^{-\rho/\beta_{d,q}}\right)\mathrm{P}_0^S, 		
		\end{equation} 
		where 
				\vspace{-0.1cm}
		\begin{equation*}
		\vspace{-0.1cm}
		\!\!\mathrm{P}_0 =  1 -  \frac{\bar{L}}{M}\!\left(\!\!\frac{\bar{L}e^{\frac{\bar{L}\beta_{d,q}}{M^2\beta_{r,q}} }}{M^2\beta_{r,q}} \mathcal{I}_0\!\left(\!\!\rho;\beta_{d,q},\!\frac{M^2}{\bar{L}}\beta_{r,q}\!\right)\!\!- \! e^{-\dfrac{\rho}{\beta_{d,q}}}\!\!\right) - e^{-\dfrac{\rho}{\beta_{d,q}}},
				\vspace{-0.1cm}
		\end{equation*} $\bar{L}\! \triangleq \!\min\{L,M\}$, and 
 $\mathcal{I}_0(x;c_1,c_2) \triangleq \stretchint{4.5ex}_{\!\!\!\!c_1}^{\infty}e^{-\left(\dfrac{x}{t}+\dfrac{t}{c_2}\right)}\text{d}t.$ 
 		\vspace{-0.1cm}
		\end{theorem}
		\begin{proof}
		Let $\mathcal{E}_s$ denote the event that $s$th IRS aligns with the OOB UE's channel. Then, $\mathcal{E}_0$ is the event that no IRS aligns with the OOB channel. We can write
		\vspace{-0.1cm}
			\begin{align}
		\!\!\!P_{\text{out},q}^{\rho} &\stackrel{(a)}{=} \mathsf{Pr}\left(|h_q|^2 \mathbbm{1}_{\{ \mathcal{E}_0\bigcup \mathcal{E}_1 \bigcup \ldots \bigcup \mathcal{E}_S\}} \leq \rho \right)\\ 
		& \stackrel{(b)}{=} \prod_{\mathfrak{s}=0}^S \mathbb{E}_{\mathbbm{1}_{\{\mathcal{E}_{\mathfrak{s}}\}}}\left[\mathsf{Pr}\left(|h_q|^2 \leq \rho \Bigg\vert\mathbbm{1}_{\{\mathcal{E}_{\mathfrak{s}}\}}\right)\right]\\
		& \stackrel{(c)}{=} \mathsf{Pr}(|h_{d,q}|^2 \leq \rho) \!\left(\mathbb{E}_{\mathbbm{1}_{\{\!\mathcal{E}_1\!\}}}\!\!\left[\mathsf{Pr}\!\left(|h_q|^2 \leq \rho\Big\vert\mathbbm{1}_{\{\!\mathcal{E}_1\!\}}\right)\right]\right)^S\!\!,\!\!  \label{eq_outage_prob_template_eqn}
		\end{align}
		where $(a)$ is because $\{\mathcal{E}_{\mathfrak{s}}\}_{\mathfrak{s}=0}^{S}$ is a set of mutually exhaustive events, $(b)$ is because these events are independent, and in $(c)$, the first product term is because under $\mathcal{E}_0$, the channel at  UE-$q$ is only due to the direct link from BS-Y, and the second term is because $\{h_q^{\mathfrak{s}}\}_{\mathfrak{s}=1}^{S}$ and $\{\mathbbm{1}_{\{\!\mathcal{E}_{\mathfrak{s}}\!\}}\}_{\mathfrak{s}=1}^S$ are sets of i.i.d. random variables. Here, $h_q^{\mathfrak{s}}$ is the component of $h_q$ via the $\mathfrak{s}$th IRS. We can now evaluate the expectation term in~\eqref{eq_outage_prob_template_eqn} as
				\vspace{-0.2cm}
		\begin{multline}
		\vspace{-0.2cm}
		\mathbb{E}_{\mathbbm{1}_{\{\!\mathcal{E}_1\!\}}}\!\!\left[\mathsf{Pr}(|h_q|^2 \leq \rho \Big\vert\mathbbm{1}_{\{\mathcal{E}_1\}})\right]  =  (1\!-\!\bar{L}/M)(\mathsf{Pr}(|h_{d,q}|^2 \leq \rho)) \\ + \bar{L}/M (\mathsf{Pr}(|h_{d,q} +  M/\sqrt{\bar{L}} \gamma^{\ell^*}_{Y,1,q}|^2 \leq \rho)).
				\vspace{-0.2cm}
		\end{multline}
		While it is easy to show that $\mathsf{Pr}(|h_{d,q}|^2 \leq \!\rho) \!=\! 1\!- e^{-\rho/\beta_{d,q}}$, the second term in the above can be computed using~\cite[Appendix C]{Yashvanth_TCOM_2023}. Finally, plugging all these values in~\eqref{eq_outage_prob_template_eqn} yields~\eqref{eqn_outage_prob_OOB_UE}.
		\end{proof}
		\vspace{-0.2cm}
		Theorem~\ref{thm_outage_prob_OOB_UE} clearly shows that the outage probability at an OOB UE decreases exponentially in $S$, even though the IRSs are randomly configured from UE-$q$'s viewpoint. Thus, a distributed IRS-aided system provides instantaneous benefits even to an OOB MO without incurring optimization costs. We next illustrate our findings via Monte Carlo simulations.\!\!
		\section{Numerical Results}\label{sec_numerical_results}
		\begin{figure*}
\vspace{-0.6cm}
\hspace{-0.25cm}
\begin{subfigure}{0.28\linewidth}
\includegraphics[width=1.05\linewidth,height=5cm]{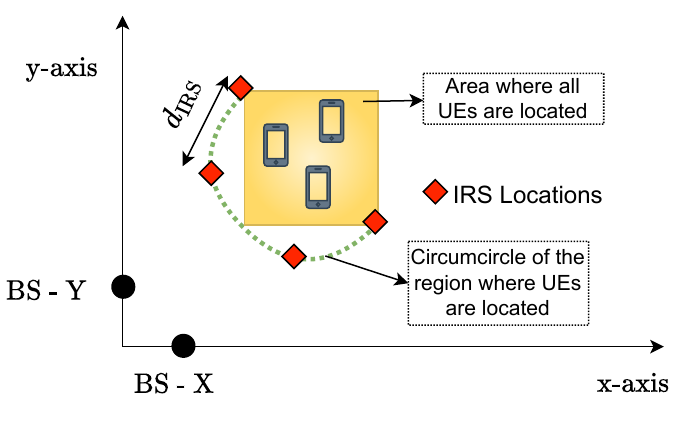}
\vspace{-0.5cm}
\caption{System setup for $S=4$.}
\label{fig_system_setup}
\end{subfigure}
\hspace{0.01cm}
\begin{subfigure}{0.39\linewidth}
\centering
\includegraphics[width=1.08\linewidth,height=5.05cm]{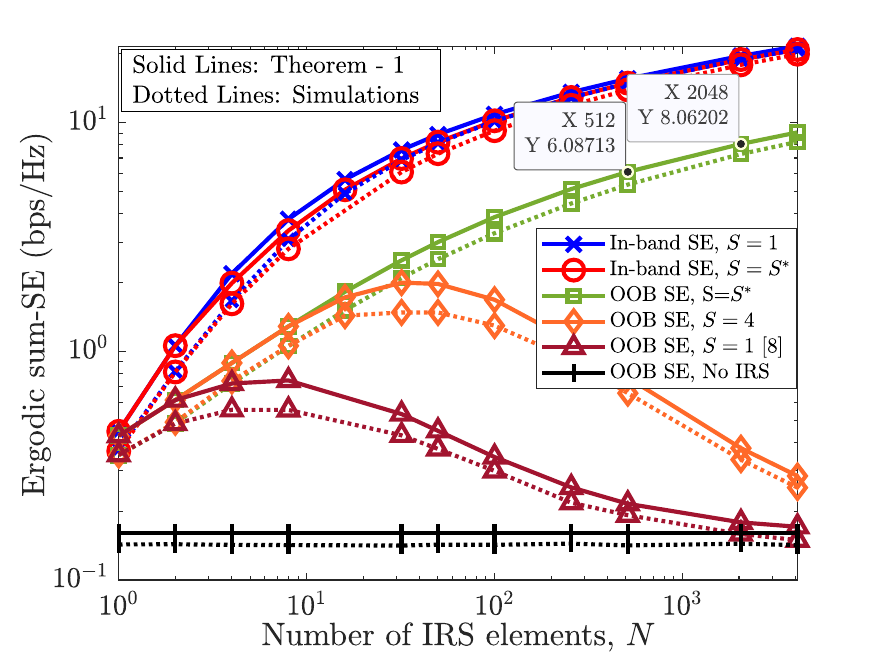}
\vspace{-0.5cm}
\caption{Ergodic sum-SE of the MOs X and Y vs. $N$.}
\label{fig_ergodic_sum_SE}
\end{subfigure}
\begin{subfigure}{0.33\linewidth}
	\includegraphics[width=1.05\linewidth,height=5cm]{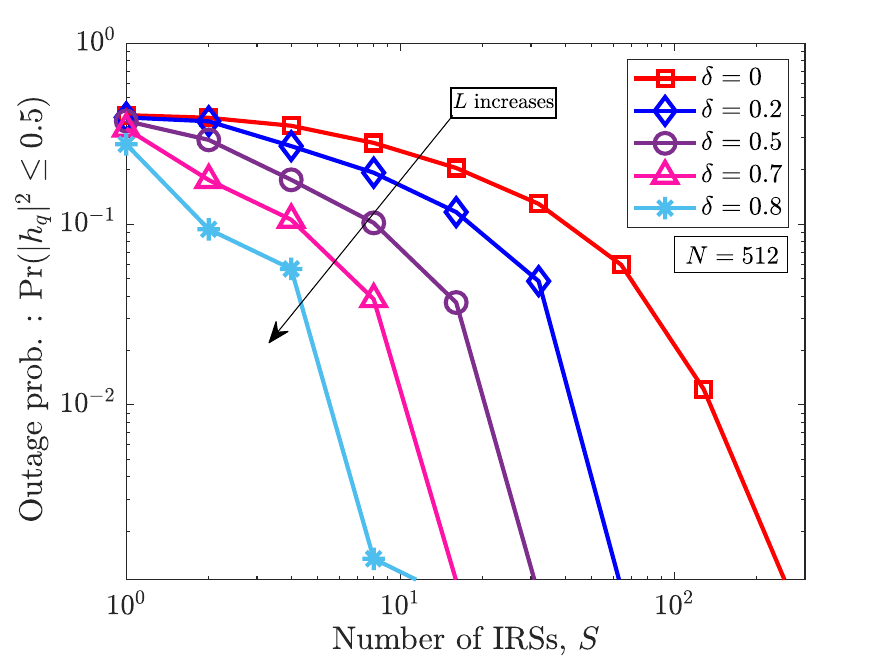}
	\vspace{-0.5cm}
	\caption{Outage probability vs. $S$ at an OOB UE.}
	\label{fig_OOB_outage_prob}
\end{subfigure}
\caption{Schematic of system setup and OOB Performance due to randomly configured distributed IRSs in mmWave frequency bands.}
\vspace{-0.6cm}
\end{figure*}
\begin{figure}
\centering
\includegraphics[width=8.8cm,height=5.7cm]{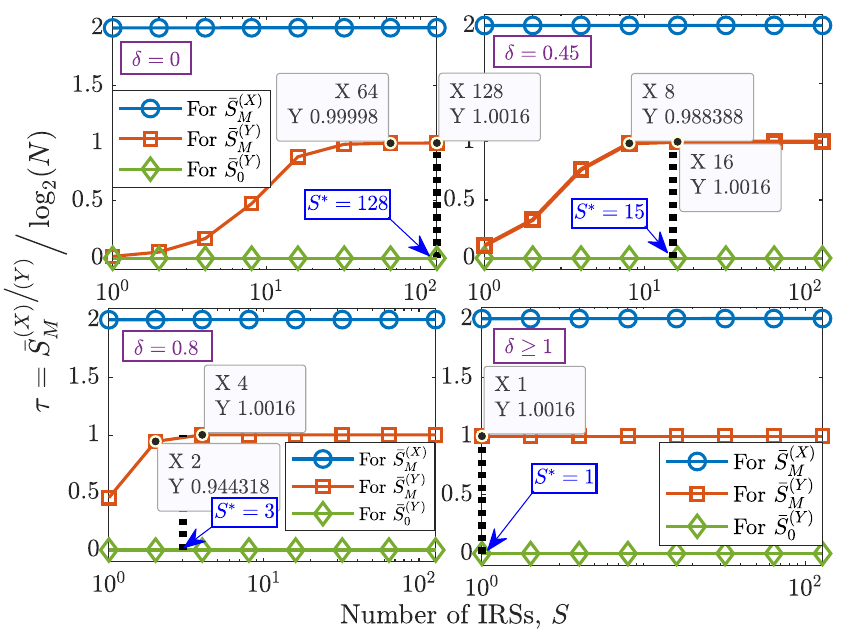}
\caption{Pre-log factor of OOB SE, $\tau$ vs. $S$ as a function of $\delta$ (or $L$).}
\label{fig_OOB_SE_scaling_in_N}
\vspace{-0.4cm}
\end{figure}
		We consider that BS-X is located at $(50,0)$, BS-Y at $(0,50)$, and the UEs in a rectangular region $\mathcal{R} \triangleq [900,1100] \times [900,1100]$. The IRSs are uniformly spaced on a semi-circular part of the circle that circumscribes $\mathcal{R}$, as shown in Fig.~\ref{fig_system_setup}. Note that our results hold true for any other choice of IRS locations and the distance between the IRSs. The path loss is modeled as $\!\beta \!= \!C_0(d_0/d)^\alpha\!$, where $C_0$ is the path loss at the reference distance $d_0$, $d$ is the node distance, and $\alpha$ is the path loss exponent which is $2,2.2$ and $4.5$ in the BS-IRS, IRS-UE, and BS-UE links. Further, BS-X and BS-Y serve $K=10$ and $Q=10$ UEs using an RR scheduler over $10,000$ time slots. \\
\indent In Fig.~\ref{fig_ergodic_sum_SE}, we plot the ergodic sum-SE of both MOs versus the total number of IRS elements, $N$, for different values of $S$. We use $C_0\left(P/\sigma^2\right) = 150$ dB and $L=2$~\cite{Yashvanth_TCOM_2023}. We first observe that the SE of the MO X is invariant to the number of IRSs in the system. Thus, the optimal SE scaling of $\mathcal{O}(2\log_2(N))$ is retained for any $S$. However, the OOB SE significantly changes with $S$. For e.g., when $S=4$, until some point, the OOB SE log-linearly increases in $N$ and then exhibits log-sub-linear growth for larger values of $N$ (in the regime where  $L < M$), i.e., the pre-log factor of OOB SE lies in $[0,1)$. The latter is because, for large $N$ and fixed small $L$, unless $S$ is reasonably large, the OOB UE does not benefit much. We then plot the OOB SE for $S=S^*$, where $S^*$ is as per Proposition~\ref{prop_OOB_SNR_linear_N}. In this case, as indicated in the data marked inside the figure, a $4\times$ increase in $N$ leads to a boost of the SE by $2$ bps/Hz. Note that $2 + \log_2(N) = 1\cdot\log_2\left((4N)\right)$. In other words, the OOB SE uniformly achieves the (maximum) scaling of $\log_2(N)$ with sufficiently many IRSs. We also compare our results with that in~\cite{Yashvanth_TCOM_2023}, which considers $S=1$ on the same plot. The OOB SE obtained in distributed IRSs is clearly better than the single IRS case. Also, the simulations match well with Theorem~\ref{thm:rate_characterization_mmwave_single_path_IB_multiple_IRS}, illustrating the accuracy of the rate laws and the Bernoulli distribution-based analysis for an OOB MO.\\
\indent In Fig.~\ref{fig_OOB_SE_scaling_in_N}, we illustrate the tightness of Proposition~\ref{prop_OOB_SNR_linear_N}. Specifically, we plot the SE pre-log factor: $\tau \triangleq \bar{S}_M^{(X)/(Y)}/\log_2(N)$ as a function of $S$. Each sub-plot represents a system with different $L$, through $\delta \triangleq \log_NL$ with $N=128$. The plot shows that the result in Proposition~\ref{prop_OOB_SNR_linear_N} is tight because, in all cases, when $S < S^*$, $\tau \in [0,1)$, i.e., the OOB SE grows only log-sub-linearly in $N$. Contrarily, for $S \geq S^*$, the OOB-SE scales as $\mathcal{O}(\log_2(N))$ for any $N,L$. Also, for the in-band UE, $\tau=2$, in line with Fig.~\ref{fig_ergodic_sum_SE}. Further,  since the OOB SE does not depend on $N$ in the absence of IRSs, we have $\tau=0$.

\indent Finally, in Fig.~\ref{fig_OOB_outage_prob}, we plot the outage probability at an OOB UE (after normalizing the channel path losses) with $\rho=0.5$ as a function of $S$ for different $L$ through $\delta = \log_N L$, and $N=512$. For a given $\delta$ (or $L$), the outage probability decreases exponentially in the number of IRSs, in line with  Theorem~\ref{thm_outage_prob_OOB_UE}. Further, for a fixed $S$, the outage probability decreases with $\delta$ as the likelihood that the randomly configured IRSs help the OOB UE increases with the number of paths.
\vspace{-0.05cm}
\section{Conclusions}
We studied the impact of distributed IRSs on the performance of an OOB MO that is oblivious to the presence of the IRSs. We showed that the maximum SE of the OOB MO grows log-linearly in the number of IRS elements; we also developed design specifications that almost surely achieve this SE. Finally, we proved that the outage probability at an OOB UE decays exponentially in the number of IRSs, free of cost. Thus, distributed IRSs help all MOs in the area. Future work includes extending our results to interference-limited scenarios in multi-cell systems.
\vspace{-0.1cm}
		\bibliographystyle{IEEEtran}
		\bibliography{IEEEabrv,References_Bibtex}
	\end{document}